\documentclass{llncs}
\usepackage{amsmath}%
\usepackage{amsfonts}%
\usepackage{amssymb}%
\usepackage{graphicx}
\usepackage[english]{babel}
\usepackage{cite}
\usepackage[all]{xy}%
\usepackage{verbatim}
\usepackage{stmaryrd}
\numberwithin{equation}{section}

\newcommand{\Sm}{\mathcal{S}}   
\newcommand{\Lm}{\mathcal{L}}  
\newcommand{\Fm}{\mathbf{Fm}}  
\newcommand{\Al}{\mathbf{A}}  
\newcommand{\BB}{\mathbf{B}}   
\newcommand{\Alg}{\mathbb{A}\mathrm{lg}}
\newcommand{\LL}{\mathcal{L}}
\newcommand{\SSS}{\mathcal{S}}
\newcommand{\la}{\langle}
\newcommand{\ra}{\rangle}

\newcommand{\s}{\mathcal{S}}
\newcommand{\Ls}{\mathcal{L}_{\s}}
\newcommand{\CRS}{\vdash_{\mathcal{S}}}
\newcommand{\aAA}{\boldsymbol{A}}
\newcommand{\Fi}{\mathrm{Fi}}

\newcommand{\bteor}{\begin{theorem}}
\newcommand{\eteor}{\end{theorem}}

\begin{document}
\mainmatter              		
\pagestyle{headings}
\title{An Abstract Algebraic Logic View on Judgment Aggregation\thanks{The research of the second and third author has been made possible by the NWO Vidi grant 016.138.314, by the NWO Aspasia grant 015.008.054, and by a Delft Technology Fellowship awarded in 2013.}}
\author{Mar\'ia Esteban\inst{1}
\and
Alessandra Palmigiano\inst{2}
\and
 Zhiguang Zhao\inst{2}}
\institute{Departament de L\`ogica, Hist\`oria i Filosofia de la Ci\`encia, Facultat de Filosofia, Universitat de Barcelona\\
\email{mariaesteban.edu@gmail.com}
\and Faculty of Technology, Policy and Management, Delft University of Technology\\
\email{\{a.palmigiano, z.zhao-3\}@tudelft.nl}}

\maketitle
\thispagestyle{empty}

\begin{abstract}
In the present paper, we propose  Abstract Algebraic Logic (AAL) as a general logical framework for Judgment Aggregation. Our main contribution is a generalization of Herzberg's algebraic approach to characterization results in  on judgment aggregation and propositional-attitude aggregation, characterizing certain Arrovian classes of aggregators as Boolean algebra and MV-algebra homomorphisms, respectively.
 The characterization result of the present paper applies to  agendas of formulas of an arbitrary {\em selfextensional} logic.
This notion comes from AAL, and encompasses a vast class of logics, of which classical, intuitionistic, modal, many-valued and relevance logics are special cases. To each selfextensional logic $\Sm$, a unique class of algebras $\Alg\Sm$ is canonically associated by the general theory of AAL. We show that for any selfextensional logic $\Sm$ such that $\Alg\Sm$ is closed under direct products, any algebra in $\Alg\Sm$ can be taken as the set of truth values on which an aggregation problem can be formulated.
In this way, judgment aggregation on agendas formalized in  classical, intuitionistic, modal, many-valued and relevance logic can be uniformly captured as special cases.
This paves the way to the systematic study of a wide array of ``realistic agendas'' made up of complex formulas, the propositional connectives of which are interpreted in ways which depart from their classical interpretation.
This is particularly interesting given that, as observed by Dietrich \cite{Di10},  nonclassical (subjunctive) interpretation of logical connectives can provide a strategy for escaping impossibility results.\\
{\em Keywords:} Judgment aggregation; Systematicity; Impossibility theorems; Abstract Algebraic Logic; Logical filter; Algebra homomorphism.\\
{\em Math. Subject Class.} 91B14; 03G27.
\end{abstract}
\section{Introduction}

\paragraph{\bf\em Social choice and judgment aggregation.} The theory of {\em social choice} is the formal study of mechanisms for collective decision making, and investigates issues of philosophical, economic, and political significance, stemming from the classical Arrovian problem of how the preferences of the members of a group can be ``fairly'' aggregated into one outcome.

In the last decades, many results appeared generalizing the original Arrovian problem, which gave rise to a research area called {\em judgment aggregation} (JA) \cite{LP09}. While the original work of Arrow \cite{Ar63} focuses on preference aggregation, this can be recognized as a special instance of the aggregation of consistent judgments, expressed by each member of a group of individuals over a given set of logically interconnected propositions (the {\em agenda}): each proposition in the agenda is either accepted or rejected by each group member, so as to satisfy certain requirements of logical consistency. Within the JA framework, the Arrovian-type {\em impossibility results} (axiomatically providing  sufficient conditions for aggregator functions to turn into degenerate rules, such as dictatorship) are obtained as consequences of {\em characterization theorems} \cite{NP02}, which provide necessary and sufficient conditions for agendas to have aggregator functions on them satisfying given axiomatic conditions.

In the same logical vein, in \cite{DiLi10}, {\em attitude aggregation theory} was introduced; this direction has been further pursued in \cite{He13}, where a characterization theorem has been given for certain many-valued propositional-attitude aggregators as MV-algebra homomorphisms. 

\paragraph{\bf\em The ultrafilter argument and its generalizations.} Methodologically, the {\em ultrafilter argument} is the tool underlying the generalizations and unifications mentioned above. It can be sketched as follows: to prove impossibility theorems for finite electorates, one shows that the axiomatic conditions on the aggregation function force the set of all decisive coalitions to be an (ultra)filter on the powerset of the electorate. If the electorate is finite, this implies that all the decisive coalitions must contain one and the same (singleton) coalition: the oligarchs (the dictator). First employed in \cite{KiSo72} for a proof of Arrow's theorem alternative to the original one, this argument was applied to obtain elegant and concise proofs of impossibility theorems also in judgment aggregation  \cite{DM10}. More recently, it gave rise to characterization theorems, e.g.\ establishing a bijective correspondence between Arrovian aggregation rules and ultrafilters on the set of individuals \cite{HE12}. Moreover,  the ultrafilter argument has been generalized by Herzberg \cite{He08} to obtain a bijective correspondence between certain judgment aggregation functions and ultraproducts of profiles, and---using the well-known correspondence between ultrafilters and Boolean homomorphisms---similar correspondences have been established between Arrovian judgment aggregators and Boolean algebra homomorphisms \cite{He10}.

\paragraph{\bf\em  Escaping impossibility via nonclassical logics.} While much research in this area explored the limits of the applicability of  Arrow-type results, at the same time the question of how to `escape impossibility' started attracting increasing interest. In \cite{Di10}, Dietrich shows that impossibility results do not apply to a wide class of realistic agendas
once propositions of the form `if $a$ then $b$' are  modelled as {\em subjunctive} implications rather than material implications. Besides its theoretical value, this result is of practical interest, given that subjunctive implication models the meaning of if-then statements in natural language more accurately than  material implication.

\paragraph{\bf\em Aim.} A natural question arising in the light of Dietrich's result is how to highlight the role played by the logic (understood both as formal language and deductive machinery) underlying the given agenda in characterization theorems for JA.

The present paper focuses on {\em Abstract Algebraic Logic}  as a natural theoretical setting for Herzberg's results \cite{He08,He13}, and the theory of {\em (fully) selfextensional logics} as the  appropriate logical framework for a nonclassical interpretation of logical connectives, in line with the approach of \cite{Di10}.

\paragraph{\bf\em Abstract Algebraic Logic and selfextensional logics.} Abstract Algebraic Logic (AAL) \cite{FoJaPi03} is a forty-year old research field in mathematical logic. It was conceived as the framework for an algebraic approach to the investigation of classes of logics. Its main goal was establishing a notion of {\em canonical algebraic semantics} uniformly holding for classes of logics, and using it to systematically investigate (metalogical) properties of logics in connection with properties of their algebraic counterparts.

{\em Selfextensionality} is the metalogical property holding of those logical systems whose associated relation of logical equivalence on formulas is a congruence of  the term algebra.  W\'ojcicki \cite{Wo79} characterized  selfextensional logics as the logics which admit a so-called {\em referential semantics} (which is a general version of the well known possible-world semantics of modal and intuitionistic logics), and in \cite{JaPa06}, a characterization was given of the particularly well behaved subclass of the fully selfextensional logics in general duality-theoretic terms. This subclass includes many well-known logics, such as classical, intuitionistic, modal, many-valued and relevance logic. These and other results in this line of research (cf.\ e.g.\ \cite{Ja05,GJP10,Es13,EsJa14}) establish a systematic connection between possible world semantics and the logical account of intensionality.

\paragraph{\bf\em Contributions.}
In the present paper, we generalize and refine Herzberg's characterization result in \cite{He13} from the MV-algebra setting to any class of algebras canonically associated with some  selfextensional logic.

In  particular, the properties of agendas are formulated independently of a specific logical signature and are slightly different than those of Herzberg's setting.  In contrast with Herzberg's characterization result, which consisted of two slightly asymmetric parts, the two propositions which yield the characterization result in the present paper (cf.\ Propositions \ref{prop1} and \ref{prop2}) are symmetric. Aggregation of propositional attitudes modeled in classical, intuitionistic, modal, \L ukasiewicz and relevance logic can be uniformly captured as special cases of the present result.
 This makes it possible to fine-tune the expressive and deductive power of the formal language of the agenda, so as to capture e.g.\ intensional or vague statements.

\paragraph{\bf\em Structure of the paper} In Section \ref{sec:preliminaries} we give preliminaries on Abstract Algebraic Logic. In Section \ref{Sec:framework} we provide the formal framework for judgment aggregation. In Section \ref{Sec:Results} we prove the correspondence result. In Section \ref{Sec:Arrow} we state the impossibility theorem for judgment aggregation as a corollary, and discuss the setting of subjunctive implication.

\section{Preliminaries on Abstract Algebraic Logic}
\label{sec:preliminaries}
The present section collects the basic concepts of Abstract Algebraic Logic that we will use in the paper.
For a general view of AAL the reader is addressed to \cite{FoJa96} and the references therein.

\subsection{General approach.} As mentioned in the introduction, in AAL,  logics are not studied in isolation, and in particular,  investigation focuses on classes of logics and their identifying {\em metalogical properties}.  Moreover, the notion of {\em consequence} rather than the notion of {\em theoremhood} is taken as basic: consequently, {\em sentential logics}, the primitive objects studied in AAL, are defined as tuples $\s = \langle\Fm, \CRS\rangle$ where $\Fm$ is the algebra of formulas of type $\Ls$ over a denumerable set of variables $Var$, and $\CRS$ is a {\em consequence relation } on (the carrier of) $\Fm$ (cf.\ Subsection \ref{subsec:logics}).

\noindent This notion encompasses logics that are defined by any sort of proof-theoretic calculus (Gentzen-style, Hilbert-style, tableaux, etc.), as well as logics arising from some classes of (set-theoretic, order-theoretic, topological, algebraic, etc.) semantic structures, and in fact it allows to treat logics independently of the way in which they have been originally introduced. Another perhaps more common approach in logic takes the notion of theoremhood as basic and consequently sees logics as sets of formulas (possibly closed under some rules of inference). This approach is easily recaptured by the notion of sentential logic adopted in AAL: Every sentential logic $\s$ is uniquely associated with the set  $Thm(\s) = \{\varphi\in Fm\ |\ \emptyset\CRS\varphi\}$ of its {\em theorems}.

\subsection{Consequence operations}
For any set $A$, a \textit{consequence operation} (or closure operator) on $A$  is a map $C: \mathcal{P}(A) \to\mathcal{P}(A)$ such that for every $X, Y \subseteq A$: (1) $X \subseteq C(X)$, (2) if $X \subseteq Y $, then $C(X) \subseteq C(Y)$ and (3) $C(C(X)) = C(X)$. The closure operator $C$ is \textit{finitary} if in addition satisfies  (4) $C(X) = \bigcup\{C(Z): Z \subseteq X, Z \text{ finite}\}$. For any consequence operation $C$ on $A$, a set $X \subseteq A$ is $C$-\textit{closed} if $C(X) = X$. Let $\mathcal{C}_{C}$ be the collection of $C$-closed subsets of $A$.

For any set $A$, a \textit{closure system} on $A$ is a collection $\mathcal{C}\subseteq\mathcal{P}(A)$ such that $A\in \mathcal{C}$, and $\mathcal{C}$ is closed under intersections of arbitrary non-empty families. A closure system is \textit{algebraic} if it is closed under unions of up-directed\footnote{
For $\langle P,\leq\rangle$ a poset, $U\subseteq P$ is  \emph{up-directed} when for any $a,b\in U$ there exists $c\in U$ such that $a,b\leq c$. } families.

For any closure operator $C$ on $A$, the collection $\mathcal{C}_{C}$ of the $C$-closed subsets of $A$ is a closure system on $A$. If $C$ is finitary, then $\mathcal{C}_{C}$ is algebraic. Any closure system $\mathcal{C}$ on $A$ defines a consequence operation $C_{\mathcal{C}}$ on $A$ by setting
$C_{\mathcal{C}}(X) = \bigcap\{Y \in \mathcal{C}:X \subseteq Y\}$
for every $X \subseteq A$. The $C_{\mathcal{C}}$-closed sets are exactly the elements of $\mathcal{C}$. Moreover, $\mathcal{C}$ is algebraic if and only if $C_{\mathcal{C}}$ is finitary.

\subsection{Logics}\label{subsec:logics}
Let $\LL$ be a propositional language (i.e.\ a set of connectives, which we will also regard as a set  of function symbols) and let $\Fm_{\LL}$ denote the algebra of formulas (or term algebra) of $\LL$ over a denumerable  set $V$ of variables. Let $Fm_{\LL}$ be the carrier of the algebra $\Fm_{\LL}$. A \textit{logic} (or deductive system) of type $\LL$ is a pair $\SSS = \la \Fm_{\LL}, \vdash_{\SSS}\ra$ such that $\vdash_{\SSS}\subseteq\mathcal{P}(Fm_{\LL})\times Fm_{\LL}$ such that the operator $C_{\vdash_{\SSS}}:\mathcal{P}(Fm_{\LL}) \to \mathcal{P}(Fm_{\LL})$ defined by $$\varphi \in C_{\vdash_{\SSS}}(\Gamma)\; \; \text{ iff } \; \; \Gamma \vdash_{\SSS} \varphi$$ is a consequence operation with the property of  \textit{invariance under substitutions}; this means that for every substitution $\sigma$ (i.e.\ for every $\LL$-homomorphism $\sigma: \Fm_{\LL}\rightarrow\Fm_{\LL}$) and for every $\Gamma \subseteq Fm_{\LL}$,

$$\text{$\sigma[C_{\vdash_{\SSS}}(\Gamma)]\subseteq C_{\vdash_{\SSS}}(\sigma[\Gamma]).$}$$

For every $\SSS$, the relation $\vdash_{\SSS}$ is the {\em consequence} or {\em entailment} relation of $\SSS$. A logic is \textit{finitary} if the consequence operation $C_{\vdash_{\SSS}}$ is finitary. Sometimes we will use the symbol $\Ls$ to refer to the propositional language of a logic $\SSS$.

The  \textit{interderivability relation} of a logic $\SSS$  is the relation $\equiv_{\SSS}$ defined by $$\varphi  \equiv_{\SSS} \psi \;\; \text{ iff } \;\; \varphi \vdash_{\SSS} \psi \text{ and } \psi \vdash_{\SSS} \varphi.$$ $\SSS$ satisfies the \textit{congruence property} if $\equiv_{\SSS}$ is a congruence of $\Fm_{\LL}$.

\subsection{Logical filters}
Let $\SSS$ be a logic of type $\LL$ and let $\aAA$ be an $\LL$-algebra (from now on, we will drop reference to the type $\LL$, and when we refer to an algebra or class of algebras in relation with $\SSS$, we will always assume that the algebra and the algebras in the class are of type $\LL$).

A subset $F \subseteq A$ is  an $\SSS$-\textit{filter} of $\aAA$ if for every $\Gamma\cup\{\varphi\}\subseteq Fm$ and every $h\in\mathrm{Hom}(\Fm_{\LL}, \aAA)$, $$\text{if } \Gamma \vdash_{\SSS}\varphi \text{ and } h[\Gamma] \subseteq F, \text{ then } h(\varphi) \in F.$$ The collection $\Fi_{\SSS}(\aAA)$ of the $\SSS$-filters of $\aAA$ is a closure system. Moreover, $\Fi_{\SSS}(\aAA)$ is an algebraic closure system if
$\SSS$ is finitary. The consequence operation associated with $\Fi_{\SSS}(\aAA)$ is denoted by $C_{\SSS}^{\aAA}$. For every $X \subseteq A$, the closed set
$C_{\SSS}^{\aAA}(X)$ is the $\SSS$-filter of $\aAA$ generated by $X$. If $\SSS$ is finitary, then $C_{\SSS}^{\aAA}$ is finitary for every algebra $\aAA$.

On the algebra of formulas  $\Fm$, the closure operator $C_{\SSS}^{\Fm}$ coincides with $C_{\vdash_{\SSS}}$ and the $C_{\SSS}^{\Fm}$-closed sets are exactly the $\SSS$-\textit{theories}; that is, the sets of formulas which are closed under the relation $\vdash_{\SSS}$.

\subsection{$\SSS$-algebras and selfextensional logics}\label{subsec:Salgebras}

One of the basic topics of AAL is how to associate in a uniform way a class of algebras with an arbitrary logic $\Sm$. According to contemporary AAL \cite{FJa09}, the canonical algebraic counterpart of $\Sm$  is the class $\Alg\Sm$, whose elements are called \emph{$\Sm$-algebras}. This class can be defined via the notion of Tarski congruence.

For any algebra $\Al$ (of the same type as $\Sm$) and any closure system $\mathcal{C}$ on $\Al$,
the \emph{Tarski congruence of $\mathcal{C}$ relative to $\Al$},  denoted by $\tilde{\mathbf{\Omega}}_{\Al}(\mathcal{C})$,
is the greatest congruence which is compatible with all $F\in \mathcal{C}$, that is, which does not relate elements of $F$ with elements which do not belong to $F$.
The Tarski congruence of the closure system consisting of all $\Sm$-theories relative to $\Fm$ is denoted by $\tilde{\mathbf{\Omega}}(\Sm)$. The quotient algebra $\Fm/\tilde{\mathbf{\Omega}}(\Sm)$ is called the \emph{Lindenbaum-Tarski algebra} of $\Sm$.

For any algebra $\Al$, we say that $\Al$ is an \emph{$\Sm$-algebra}  (cf.\ \cite[Definition 2.16]{FJa09}) if the Tarski congruence of $\Fi_{\SSS}(\aAA)$ relative to $\Al$ is the identity. It is well-known (cf.\ \cite[Theorem 2.23]{FJa09} and ensuing discussion) that $\Alg\Sm$ is closed under direct products. Moreover, for any logic $\Sm$, the Lindenbaum-Tarski algebra is an $\Sm$-algebra (see page 36 in \cite{FJa09}).

A logic $\Sm$ is \emph{selfextensional} when the relation of \emph{logical equivalence} between formulas
$$\varphi\ {\equiv_{\Sm}}\ \psi\quad\quad\mbox{ iff }\quad\quad\varphi\ {\vdash_{\Sm}}\ \psi\ \ \mbox{ and }\ \ \psi\ {\vdash_{\Sm}}\ \varphi$$
is a congruence relation of the formula algebra $\Fm$.
An equivalent definition of selfextensionality (see page 48 in \cite{FJa09}) is given as follows: $\Sm$ is selfextensional iff  the Tarski congruence $\tilde{\mathbf{\Omega}}(\Sm)$ and the relation of logical equivalence $\equiv_{\Sm}$ coincide. In such case the Lindenbaum-Tarski algebra reduces to ${\Fm/\equiv_{\Sm}}$. Examples of selfextensional logics besides classical propositional logic are intuitionistic logic, positive modal logic \cite{CeJa99}, the $\{\land, \lor\}$-fragment of classical propositional logic,  Belnap's four-valued logic \cite{Be77} the local consequence relation associated with Kripke frames, and the (order-induced) consequence relation associated with MV-algebras and defined by ``preserving degrees of truth'' (cf.\ \cite{Fo03}). Examples of non-selfextensional logics include linear logic, the (1-induced) consequence relation associated with MV-algebras and defined by ``preserving absolute truth'' (cf.\ \cite{Fo03}), and the global consequence relation associated with Kripke frames.

From now on we assume that $\Sm$ is a selfextensional logic and $\BB\in\Alg\Sm$. For any formula $\varphi\in Fm$, we say that $\varphi$ is \emph{provably equivalent} to a variable iff there exist a variable $x$ such that $\varphi\equiv x$.

\section{Formal framework}
\label{Sec:framework}
In the present section, we generalize Herzberg's algebraic framework for aggregation theory from MV-propositional attitudes to $\Sm$-propositional attitudes, where $\Sm$ is an arbitrary selfextensional logic. Our conventional notation is similar to \cite{He13}. Let $\Lm$ be a logical language which contains countably many connectives, each of which has arity at most $n$, and let $Fm$ be the collection of $\Lm$-formulas.

\subsection{The agenda}\label{subsec:agenda}

The \emph{agenda} will be given by a set of formulas $X\subseteq \Fm$. Let $\bar{X}$ denote the closure of $X$ under the connectives of the language. Notice that for any constant $c\in\Lm$, we have $c\in \bar{X}$.

We want the agenda to contain a sufficiently rich collection of formulas.
In the classical case, it is customary to assume that the agenda contains at least two propositional variables.
In our general framework, this translates in the requirement that the agenda contains at least $n$ formulas that `behave' like propositional variables, in the sense that their interpretation is not constrained by the interpretation of any other formula in the agenda.

We could just assume that the agenda contains at least $n$ different propositional variables, but we will deal with a slightly more general situation, namely, we assume that the agenda is $n$-pseudo-rich:

\begin{definition}\label{def:pseudo:rich}
An agenda is $n$-\emph{pseudo-rich}, if it contains at least $n$ formulas $\{\delta_{1},\dots,\delta_{n}\}$ such that each $\delta_{i}$ is provably equivalent to $x_i$ for some set $\{x_1,\ldots, x_n\}$ of pairwise different variables.
\end{definition}

\subsection{Attitude functions, profiles and attitude aggregators}

An \emph{attitude function} is a function ${A\in \BB^{X}}$ which assigns each formula in the agenda to an element of the algebra $\BB$.

The \emph{electorate} will be given by some (finite or infinite) set $N$. Each $i\in N$ is called an \emph{individual}.

An \emph{attitude profile} is an $N$-sequence of attitude functions, i.e.\ $\vec{A}\in (\BB^{X})^{N}$.
For each $\varphi\in X$, we denote the $N$-sequence $\{A_{i}(\varphi)\}_{i\in N}\in \BB^{N}$ by $\vec{A}(\varphi)$.

An attitude aggregator is a function which maps each profile of individual attitude functions in some domain to a collective attitude function, interpreted as the set of preferences of the electorate as a whole.
Formally, an \emph{attitude aggregator} is a partial map $F:(\BB^{X})^{N} \nrightarrow \BB^{X}$.

\subsection{Rationality}
Let the agenda contain formulas $\varphi_{1},\dots,\varphi_{m},g(\varphi_{1},\dots,\varphi_{m})\in X$, where $g\in\Lm$ is an $m$-ary connective of the language and $m\leq n$.
Among all attitude functions $A\in\BB^{X}$, those for which it holds that $A(g(\varphi_{1},\dots,\varphi_{n}))=g^{\BB}(A(\varphi_{1}),\dots,A(\varphi_{n}))$ are of special interest. In general, we will focus on attitude functions which are `consistent' with the logic $\Sm$ in the following sense.

We say that an attitude function $A\in \BB^{X}$ is \emph{rational} if it can be extended to a homomorphism $\bar{A}:\Fm_{/\equiv}\longrightarrow \BB$ of $\SSS$-algebras. In particular, if $A$ is rational, then it can be uniquely extended to $\bar{X}$, and we will implicitly use this fact in what follows.

We say that a profile $\vec{A}\in (\BB^{X})^{N}$ is \emph{rational} if $A_{i}$ is a rational attitude function for each $i\in N$.

We say that an attitude aggregator $F:(\BB^{X})^{N} \nrightarrow \BB^{X}$ is \emph{rational} if for all rational profiles $\vec{A}\in dom(F)$ in its domain, $F(\vec{A})$ is a rational attitude function. Moreover, we say that $F$ is \emph{universal} if $\vec{A}\in dom(F)$ for any rational profile $\vec{A}$. In other words, an aggregator is universal whenever its domain contains all rational profiles, and it is rational whenever it gives a rational output provided a rational input.

\subsection{Decision criteria and systematicity}

A \emph{decision criterion} for $F$ is a partial map $f:\BB^{N}\nrightarrow \BB$ such that for all $\vec{A}\in dom (F)$ and all $\varphi\in X$,
\begin{equation}\label{eq1}
F(\vec{A})(\varphi)= f(\vec{A}(\varphi)).
\end{equation}

As observed by Herzberg \cite{He13}, an aggregator is independent if the aggregate attitude towards any proposition $\varphi$ does not depend on the individuals attitudes towards propositions other than $\varphi$:

An aggregator $F$ is \emph{independent} if there exists some map $g:\BB^{N}\times X\nrightarrow \BB$ such that for all $\vec{A}\in dom(F)$, the following diagram commutes (whenever the partial maps are defined):

$$\xymatrix{
{X} && {\BB^{N}\times X} \ar[ll];[]^*+[o]{\vec{A}, id_{X}} \\
&& {\BB}  \ar[ull];[]_*+[o]{F(\vec{A})}  \ar[u];[]^*+[o]{g}}
$$

An aggregator $F$ is \emph{systematic} if there exists some \emph{decision criterion} $f$ for $F$, i.e. there exists some map $f:\BB^{N}\nrightarrow \BB$ such that for all $\vec{A}\in dom(F)$, the following diagram commutes (whenever the partial maps are defined):

$$\xymatrix{
{X} & {\BB^{N}} \ar[l];[]^*+[o]{\vec{A}} \\
& {\BB}  \ar[ul];[]_*+[o]{F(\vec{A})}  \ar[u];[]^*+[o]{f}}
$$

Systematic aggregation is a special case of independent aggregation, in which the output of $g$ does not depend on the input in the second coordinate.
Thus, $g$ is reduced to a decision criterion $f:\BB^{N}\nrightarrow \BB$.

An aggregator $F$ is \emph{strongly systematic} if there exists some decision criterion $f$ for $F$, such that for all $\vec{A}\in dom(F)$, the following diagram commutes (whenever the partial maps are defined):
	
$$\xymatrix{
{\bar{X}} & {\BB^{N}} \ar[l];[]^*+[o]{\vec{A}} \\
& {\BB}  \ar[ul];[]_*+[o]{F(\vec{A})}  \ar[u];[]^*+[o]{f}}
$$

Notice that the diagram above differs from the previous one in that the agenda $X$ is now replaced by its closure $\bar{X}$ under the connectives of the language. If $X$ is closed under the operations in $\Ls$, then systematicity and strong systematicity coincide.

A formula $\varphi\in\Fm$ is \emph{strictly contingent} if for all $a\in\BB$ there exists some homomorphism $v:\Fm\rightarrow\BB$ such that $v(\varphi)=a$.
Notice that for any $n\geq 1$, any $n$-pseudo rich agenda (cf.\ Definition \ref{def:pseudo:rich}) always contains a strictly contingent formula.
Moreover, if the agenda contains some strictly contingent formula $\varphi$, then any universal systematic attitude aggregator $F$ has a unique decision criterion (cf.\ \cite[Remark 3.5]{He13}).\label{page:unique:decision:criterion}

Before moving on to the main section, we mention four definitions which appear in Herzberg's paper, namely that of \emph{Paretian} attitude aggregator (cf.\ \cite[Definition 3.7]{He13}), \emph{complex} and \emph{rich} agendas (cf.\ \cite[Definition 3.8]{He13}), and \emph{strongly systematizable} aggregators (cf.\ \cite[Definition 3.9]{He13}). Unlike the previous ones, these definitions rely on the specific MV-signature, and thus do not have a natural counterpart in the present, vastly more general setting. However, as we will see, our main result can be formulated independently of these definitions.
Moreover, a generalization of the Pareto condition follows from the assumptions of  $F$ being universal, rational and strongly systematic, as then it holds that for any constant $c\in\Ls$, and any $\varphi\in\Fm$, if $A_i(\varphi)=c$ for all $i\in N$, then $F(\vec{A})(\varphi)=c$.

\section{Results}
\label{Sec:Results}
\begin{lemma}
Let $X$ be an $n$-pseudo-rich agenda, $m\leq n$, $g\in \Lm$ be an $m$-ary connective and $a_{1},\dots,a_{m}\in \BB$.
Then there exist formulas $\delta_{1},\dots,\delta_{m}\in X$ in the agenda and a rational attitude function $A:X\longrightarrow \BB$ such that $A(\delta_{j})=a_{j}$ for each $j\in \{1,\dots,m\}$.
\end{lemma}
\begin{proof}
As the agenda is $n$-pseudo-rich, there are formulas $\delta_{1},\dots,\delta_{m}\in X$ each of which is provably equivalent to a different variable $x_i$. Notice that this implies that the formulas $\delta_1, \ldots, \delta_m$ are not pairwise interderivable. So the $\equiv$-equivalence cells $[\delta_{1}],\dots, [\delta_{m}]$ are pairwise different, and moreover there exists a valuation $v:\Fm/{\equiv}\longrightarrow \BB$ such that $v(\delta_{i})=a_{i}$ for all $i\in \{1,\dots,m\}$.
Let $A:=v\circ\pi_{\upharpoonright X}$, where $\pi_{\upharpoonright X}:X\rightarrow\Fm/{\equiv}$ is the restriction of the canonical projection $\pi:\Fm\rightarrow\Fm/{\equiv}$ to $X$. Then clearly $A:X\rightarrow\BB$ is the required rational attitude function.
\end{proof}

\begin{lemma}\label{lemma2}
Let $X$ be an $n$-pseudo-rich agenda, $m\leq n$, $g\in \Lm$ be an $m$-ary connective and $\vec{a}_{1},\dots,\vec{a}_{m}\in \BB^{N}$.
Then there exist formulas $\delta_{1},\dots,\delta_{m}\in X$ in the agenda and a rational attitude profile $\vec{A}:X\longrightarrow \BB^{N}$ such that $\vec{A}(\delta_{j})=\vec{a}_{j}$ for each $j\in \{1,\dots,m\}$.
\end{lemma}
\begin{proof}
As the agenda is $n$-pseudo-rich, there are formulas $\delta_{1},\dots,\delta_{m}\in X$ each of which is provably equivalent to a different variable $x_i$. By previous lemma, for each $i\in N$, there exists a rational attitude function $A_{i}:X\longrightarrow \BB$ such that $A_{i}(\delta_{j})=\vec{a}_{j}(i)$ for each $j\in\{1,\dots,m\}$.
Thus it is easy to check that the sequence of attitudes $\vec{A}:=\{A_{i}\}_{i\in N}$
 is  a rational profile such that $\vec{A}(\delta_{j})=\vec{a}_{j}$ for each $j\in\{1,\dots,m\}$.
\end{proof}

Recall that given that $X$ is $n$-pseudo rich, there exists a unique decision criterion for any strongly systematic attitude aggregator $F$ (cf.\ page \pageref{page:unique:decision:criterion}).

\begin{proposition}\label{prop1}
Let $F$ be a rational, universal and strongly systematic attitude aggregator. Then the decision criterion of $F$ is a homomorphism of $\SSS$-algebras.
\end{proposition}
\begin{proof}

By the strong systematicity of $F$, there exists a decision criterion $f:\BB^{N}\longrightarrow \BB$ of $F$. Let us show that under the assumptions of the proposition, $f$ is a homomorphism of $\SSS$-algebras, by showing that $f(c^{\BB^{N}})=c^{\BB}$ for each constant $c\in\Lm$, and that, for each $m$-ary connective $g\in\Lm$, and any $\vec{a}_{1},\dots,\vec{a}_{m}\in \BB^{N}$,

$$f(g^{\BB^{N}}(\vec{a}_{1},\dots,\vec{a}_{m}))=g^{\BB}(f(\vec{a}_{1}),\dots,f(\vec{a}_{m})).$$

Let $c\in \Lm$ be a constant and let $\vec{A}$ be a rational profile. By definition $c\in \bar{X}$.
Since $\vec{A}$ is rational, it can be extended to $\bar{X}$ so that $\vec{A}(c)=c^{\BB^{N}}$.
Moreover, as $F(\vec{A})$ is also rational, it can also be extended to $\bar{X}$ so that $F(\vec{A})(c)=c^{\BB}$.
Therefore, by $F$ being strongly systematic, we get:
$$f(c^{\BB^{N}})=f(\vec{A}(c))=F(\vec{A})(c)=c^{\BB}.$$

Let $g\in\Lm$ be an $m$-ary connective, where recall that $m\leq n$ and  let $\vec{a}_{1},\dots,\vec{a}_{m}\in \BB^{N}$.
By Lemma \ref{lemma2} and the $n$-pseudo-richness of the agenda, there are formulas $\delta_{1},\dots,\delta_{m}\in X$ and a rational profile $\vec{A}:X\longrightarrow \BB$ such that
\begin{equation}\label{equ:1}
\vec{A}(\delta_{j})=\vec{a}_{j} \,\,\, \text{ for each } j\in\{1,\dots,m\}.
\end{equation}

Notice that for each $i\in N$, $A_{i}$ being rational implies that it can be extended to $\bar{X}$ so that $A_{i}(g(\delta_{1},\dots,\delta_{m}))=g^{\BB}(A_{i}(\delta_{1}),\dots,A_{i}(\delta_{m}))=g^{\BB}(\vec{a}_{1}(i),\dots,\vec{a}_{m}(i))$.
Therefore, by the definition of operation the $g$ in the product algebra $\BB^{N}\in\Alg\Sm$ we have
\begin{equation}\label{equ:2}
\vec{A}(g(\delta_{1},\dots,\delta_{m}))=g^{\BB^{N}}(\vec{a}_{1},\dots,\vec{a}_{m}).
\end{equation}

By the assumption of $F$ being universal, it follows that $\vec{A}\in dom(F)$. Moreover, since $F$ is rational, $F(\vec{A})$ is a rational attitude function, hence it can be extended to $\bar{X}$ so that
\begin{equation}\label{equ:3}
F(\vec{A})(g(\delta_{1},\dots,\delta_{m}))=g^{\BB}(F(\vec{A})(\delta_{1}),\dots,F(\vec{A})(\delta_{m})).
\end{equation}

Finally, since $F$ is strongly systematic,
\begin{equation}\label{equ:4}
F(\vec{A})(\delta_{j})=f(\vec{A}(\delta_{j})) \,\,\, \text{ for each } j\in\{1,\dots,m\}, \text{ and }
\end{equation}
\begin{equation}\label{equ:5}
F(\vec{A})(g(\delta_{1},\dots,\delta_{m}))= f(\vec{A}(g(\delta_{1},\dots,\delta_{m}))).
\end{equation}

Hence,

\begin{center}
\begin{tabular}{r c l l}
$f(g^{\BB^{N}}(\vec{a}_{1},\dots,\vec{a}_{m}))$&=&$f(\vec{A}(g(\delta_{1},\dots,\delta_{m})))$ & \eqref{equ:2}\\
&=&$F(\vec{A})(g(\delta_{1},\dots,\delta_{m}))$& \eqref{equ:5}\\
&=&$g^{\BB}(F(\vec{A})(\delta_{1}),\dots,F(\vec{A})(\delta_{m}))$& \eqref{equ:3}\\
&=&$g^{\BB}(f(\vec{A}(\delta_{1})),\dots,f(\vec{A}(\delta_{m})))$& \eqref{equ:4}\\
&=&$g^{\BB}(f(\vec{a}_{1}),\dots,f(\vec{a}_{m})))$,& \eqref{equ:1}
\end{tabular}
\end{center}

as required.
\end{proof}
\begin{proposition}\label{prop2}
Let $f:\BB^{N}\nrightarrow \BB$ be a homomorphism of $\SSS$-algebras. Then the function $F:(\BB^{X})^{N} \nrightarrow \BB^{X}$, defined for any rational profile $\vec{A}$ and any $\varphi\in X$ by the following assignment: $$F(\vec{A})(\varphi)=f(\vec{A}(\varphi)),$$ is a rational, universal and strongly systematic attitude aggregator.
\end{proposition}
\begin{proof}
By definition $F$ is a universal aggregator, and moreover its domain coincides with the set of all rational profiles.

Let $\vec{A}$ be a rational profile. Then it can be extended to a homomorphism $v:{\Fm/\equiv}\longrightarrow \BB^{N}$. Hence, $f\circ v$ is also a homomorphism. Let us show that the restriction of $f\circ v$ to the agenda $X$ coincides with $F(\vec{A})$. Indeed, for any $\varphi\in X$ we have $$f(v(\varphi))=f(\vec{A}(\varphi))=F(\vec{A})(\varphi),$$ as required. This shows that $F(\vec{A})$ is a rational attitude function, and hence $F$ is a rational attitude aggregator.

It immediately follows from the definition that $F$ is systematic. Let $\vec{A}\in dom(F)$. By definition $\vec{A}$ is rational, so it can be uniquely extended to $\bar{X}$. Therefore, by a similar argument as in previous paragraph, we can uniquely extend $F(\vec{A})$ to $\bar{X}$ so that for any $\varphi\in\bar{X}$, $F(\vec{A})(\varphi)=f(\vec{A}(\varphi))$. This shows that $F$ is also strongly systematic, which finishes the proof.
\end{proof}
Finally, the conclusion of the following corollary expresses a property which is a generalization of the Pareto condition (cf.\ \cite[Definition 3.7]{He13}).

\begin{corollary}
If $F$ is universal, rational and strongly systematic, then for any constant $c\in\Ls$ and any $\varphi\in\Fm$, if $A_i(\varphi)=c^{\BB}$ for all $i\in N$, then $F(\vec{A})(\varphi)=c^{\BB}$.
\end{corollary}
\begin{proof}
Let $c\in\Ls$ and $\varphi\in\Fm$.
Notice that by definition of the product algebra, the sequence $\{c^{\BB}\}_{i\in N}$ is precisely $c^{\BB^{N}}$.
If $A_i(\varphi)=c^{\BB}$ for all $i\in N$, i.e.\ $\vec{A}(\varphi)=c^{\BB^{N}}$,
then by Proposition \ref{prop1},
$F(\vec{A})(\varphi)=f(\vec{A}(\varphi))=f(c^{\BB^{N}})=c^{\BB}$, as required.\end{proof}

\section{Applications}
\label{Sec:Arrow}
In the present section, we show how the setting in the present paper relates to existing settings in the literature. 
\subsection{Arrow-type impossibility theorem for judgment aggregation}
Let $\Sm$ be the classical propositional logic. Its algebraic counterpart $\Alg\Sm=\mathbb{BA}$ is the variety of Boolean algebras. Let $\Lm=\{\neg, \vee\}$ be its language (the connectives $\wedge,\to,\leftrightarrow$ are definable from the primitive ones). Let $\BB=\mathbf{2}$ be the two-element Boolean algebra. Let $X\subseteq Fm_{\Lm}$ be a 2-pseudo-rich agenda.

By Propositions \ref{prop1} and \ref{prop2}, for every electorate $N$, there exists a bijection between rational, universal and strongly systematic attitude aggregators $F:(\mathbf{2}^{X})^{N}\longrightarrow \mathbf{2}^{X}$
\footnote{Note that in this case an alternative presentation of $F$ is $F:\mathcal{P}(X)^{N}\longrightarrow \mathcal{P}(X)$, which is the standard one.}
and Boolean homomorphisms $f:\mathbf{2}^{N}\longrightarrow \mathbf{2}$.

Recall that there is a bijective correspondence between Boolean homomorphisms $f:\mathbf{2}^{N}\longrightarrow\mathbf{2}$ and ultrafilters of $\mathbf{2}^N$. Moreover, if $N$ is finite, every ultrafilter of $\mathbf{2}^N$ is principal. In this case, a decision criterion corresponds to an ultrafilter exactly when it is dictatorial.

\subsection{An example of a natural interpretation for subjunctive implication}

In \cite{Di10}, Dietrich argues that, in order to reflect the meaning of connection rules (i.e.\  formulas of the form $p\rightarrow q$ or $p\leftrightarrow q $ such that $p$ and $q$ are conjunctions of atomic propositions or negated atomic propositions) as they are understood and used in natural language, the connective $\rightarrow$ should be interpreted subjunctively.That is, the formula $p\rightarrow q$ should not be understood as a statement about the actual
world, but about whether $q$ holds in hypothetical world(s) where $p$ holds,  depends on $q$'s truth value in possibly non-actual worlds.
Dietrich proposes that, in the context of connection rules, any such implication should satisfy the following conditions:
\begin{itemize}
\item[(a)] for any atomic propositions $p$ and $q$,  $p \rightarrow q$ is inconsistent
with $\{p, \neg q\}$ but consistent with each of  $\{p, q\}$  $\{\neg p, q\}$  $\{\neg p, \neg q\}$;

\item[(b)] for any  atomic propositions $p$ and $q$,  $\neg(p \rightarrow q)$ is consistent with each of $\{p, \neg q\}$,  $\{p, q\}$,  $\{\neg p, q\}$ and  $\{\neg p, \neg q\}$.
\end{itemize}
Clearly, the classical interpretation of $p\rightarrow q$ as $\neg p\vee q$ satisfies only condition (a) but not (b).  The subjunctive interpretation of $\rightarrow$ has been formalised in various settings based on possible-worlds semantics. One such setting, which is different from the one adopted by Dietrich's, is given by Boolean algebras with operators (BAOs). These are Boolean algebras endowed with an additional unary operation $\Box$ satisfying the identities $\Box 1 = 1$ and $\Box(x\wedge y) = \Box x\wedge \Box y$. Let us further restrict ourselves to the class of BAOs such that the inequality $\Box x\leq x$ is valid. This class coincides with $\Alg\SSS$, where $\SSS$ is the normal modal logic {\bf T} with the so-called local consequence relation. It is well known that {\bf T} is selfextensional and is complete w.r.t.\ the class of reflexive Kripke frames.  In this setting, let us stipulate that $p\rightarrow q$ is interpreted as $\Box(\neg p\vee q)$. 

It is easy to see that this interpretation satisfies both conditions (a) and (b).
To show that $p\rightarrow q$ is inconsistent with $\{p,\neg q\}$, observe that $\Box(\neg p\vee q)\wedge p\wedge \neg q\leq (\neg p\vee q)\wedge p\wedge \neg q = (\neg p\wedge (p\wedge \neg q))\vee (q\wedge (p\wedge \neg q))=\bot\vee\bot=\bot$. 

To show that $p\rightarrow q$ is consistent with $\{p,q\}$,
 consider the two-element BAO s.t.\ $\Box 1=1$ and $\Box 0=0$. The assignment mapping $p$ and $q$ to $1$ witnesses the required consistency statement. The remaining part of the proof is similar and hence is omitted.

Clearly, the characterization theorem given by Propositions \ref{prop1} and \ref{prop2} applies also to this setting. However, the main interest of this setting is given by the possibility theorems. It would be a worthwile future research direction to explore the interplay and the scope of these results.

\bibliographystyle{abbrv}
\bibliography{SocialChoice}{}

\begin{thebibliography}{10}

\bibitem{Ar63}
K.~J. Arrow.
\newblock {\em Social choice and individual values}, volume~12.
\newblock John Wiley, New York, 2nd edition edition, 1963.

\bibitem{Be77}
N.~D. Belnap~Jr.
\newblock A useful four-valued logic.
\newblock In {\em Modern uses of multiple-valued logic}, pages 5--37. Springer,
  1977.

\bibitem{CeJa99}
S.~Celani and R.~Jansana.
\newblock Priestley duality, a sahlqvist theorem and a goldblatt-thomason
  theorem for positive modal logic.
\newblock {\em Logic Journal of IGPL}, 7(6):683--715, 1999.

\bibitem{Di10}
F.~Dietrich.
\newblock {The possibility of judgment aggregation on agendas with subjunctive
  implications}.
\newblock {\em Journal of Economic Theory}, 145(2):603--638, March 2010.

\bibitem{DM10}
F.~Dietrich and P.~Mongin.
\newblock The premiss-based approach to judgment aggregation.
\newblock {\em Journal of Economic Theory}, 145(2):562--582, 2010.

\bibitem{Es13}
M.~Esteban.
\newblock {\em Duality Theory and Abstract Algebraic Logic}.
\newblock PhD thesis, Universitat de Barcelona,
  http://www.tdx.cat/handle/10803/125336, November 2013.

\bibitem{EsJa14}
M.~Esteban and R.~Jansana.
\newblock Priestley style duality for filter distributive congruential logics.
\newblock 2015.

\bibitem{FoJa96}
J.~Font and R.~Jansana.
\newblock {\em A general algebraic semantics for sentential logics}.
\newblock Lecture notes in logic. Springer-Verlag, 1996.

\bibitem{FoJaPi03}
J.~Font, R.~Jansana, and D.~Pigozzi.
\newblock A survey of abstract algebraic logic.
\newblock {\em Studia Logica}, 74(1/2):13--97, Jun. -- Jul. 2003.

\bibitem{Fo03}
J.~M. Font.
\newblock Beyond two.
\newblock chapter An Abstract Algebraic Logic View of Some Multiple-valued
  Logics, pages 25--57. Physica-Verlag GmbH, Heidelberg, Germany, Germany,
  2003.

\bibitem{FJa09}
J.~M. Font and R.~Jansana.
\newblock {\em A General Algebraic Semantics for Sentential Logics}, volume~7
  of {\em Lectures Notes in Logic}.
\newblock The Association for Symbolic Logic, Ithaca, N.Y., second edition,
  2009.

\bibitem{GJP10}
M.~Gehrke, R.~Jansana, and A.~Palmigiano.
\newblock Canonical extensions for congruential logics with the deduction
  theorem.
\newblock {\em Annals of Pure and Applied Logic}, 161(12):1502--1519, 2010.

\bibitem{He08}
F.~Herzberg.
\newblock Judgment aggregation functions and ultraproducts.
\newblock Institute of Mathematical Economics, University of Bielefeld, 2008.

\bibitem{He10}
F.~Herzberg.
\newblock Judgment aggregators and boolean algebra homomorphisms.
\newblock {\em Journal of Mathematical Economics}, 46(1):132 -- 140, 2010.

\bibitem{He13}
F.~Herzberg.
\newblock {Universal algebra for general aggregation theory: Many-valued
  propositional-attitude aggregators as MV-homomorphisms}.
\newblock {\em Journal of Logic and Computation}, 2013.

\bibitem{HE12}
F.~Herzberg and D.~Eckert.
\newblock Impossibility results for infinite-electorate abstract aggregation
  rules.
\newblock {\em Journal of Philosophical Logic}, 41:273---286, 2012.

\bibitem{Ja05}
R.~Jansana.
\newblock Selfextensional logics with implication.
\newblock In {\em Logica Universalis}, pages 65--88. Birkh{\"a}user Basel,
  2005.

\bibitem{JaPa06}
R.~Jansana and A.~Palmigiano.
\newblock Referential semantics: duality and applications.
\newblock {\em Reports on Mathematical Logic}, 41:63--93, 2006.

\bibitem{KiSo72}
A.~P. Kirman and D.~Sondermann.
\newblock Arrow's theorem, many agents, and invisible dictators.
\newblock {\em Journal of Economic Theory}, 5(2):267--277, 1972.

\bibitem{DiLi10}
C.~List and F.~Dietrich.
\newblock The aggregation of propositional attitudes: towards a general theory.
\newblock In {\em Oxford Studies in Epistemology}, volume~3, pages 215--234.
  Oxford University Press, 2010.

\bibitem{LP09}
C.~List and B.~Polak.
\newblock Introduction to judgment aggregation.
\newblock {\em Journal of Economic Theory}, 145(2):441 -- 466, 2010.
\newblock Judgment Aggregation.

\bibitem{NP02}
K.~Nehring and C.~Puppe.
\newblock Strategy-proof social choice on single-peaked domains: Possibility,
  impossibility and the space between.
\newblock University of California at Davis, 2002.

\bibitem{Wo79}
R.~W{\'o}jcicki.
\newblock Referential matrix semantics for propositional calculi.
\newblock {\em Bulletin of the Section of Logic}, 8(4):170--176, 1979.

\end{thebibliography}

\end{document}